\documentclass[letterpaper, 11pt]{article}

\usepackage{graphicx}

\usepackage{amsmath}

\newcommand{\ARightarrow}{\stackrel{a}{\Rightarrow}}

\newcommand{\comment}[1]{}

\newcommand{\shortdividerline}{
\begin{center} \line(1,0){150} \end{center}
}
\newcommand{\dividerline}{\begin{center}\hrule\end{center}}

\usepackage{graphicx}
\usepackage{latexsym}

\newcommand{\scriptf}{\mathcal{F}}
\newcommand{\scripte}{\mathcal{E}}
\newcommand{\scriptv}{\mathcal{V}}

\newenvironment{proof}{\paragraph{\bf Proof:}}{\hspace*{\fill}\(\Box\)}
\newtheorem{theorem}{Theorem}

\newtheorem{claim}{Claim}
\newtheorem{corollary}{Corollary}
\newtheorem{definition}{Definition}

\newtheorem{lemma}{Lemma}

\def\noflash#1{\setbox0=\hbox{#1}\hbox to 1\wd0{\hfill}}

\begin{document}

\title{Iterative Approximate Byzantine Consensus\\ in Arbitrary Directed Graphs - Part II:\\ Synchronous and Asynchronous Systems
 \footnote{\normalsize This research is supported
in part by National
Science Foundation award CNS 1059540 and Army
Research Office grant W-911-NF-0710287. Any opinions, findings, and conclusions or recommendations expressed here are those of the authors and do not
necessarily reflect the views of the funding agencies or the U.S. government.}}

\author{Nitin Vaidya$^{1,3}$, Lewis Tseng$^{2,3}$, and Guanfeng Liang$^{1,3}$\\~\\
 \normalsize $^1$ Department of Electrical and Computer Engineering, \\
 \normalsize $^2$ Department of Computer Science,
 and\\ \normalsize $^3$ Coordinated Science Laboratory\\ \normalsize University of Illinois at Urbana-Champaign\\ \normalsize Email: \{nhv, ltseng3, gliang2\}@illinois.edu~\\~\\Technical Report}

\date{February 27, 2012}

\maketitle

\thispagestyle{empty}

\newpage

\section{Introduction}
\label{s_intro}

This report contains two related sets of results with different assumptions on synchrony. The first part is about iterative algorithms in synchronous systems. Following our previous work on synchronous iterative approximate Byzantine consensus (IABC) algorithms \cite{IBA_sync}, we provide a more intuitive {\em tight} necessary and sufficient condition for the existence of such algorithms in synchronous networks\footnote{With a slight abuse of terminology, we use ``systems'' and ``networks'' interchangeably in this report.}. We believe this condition and the results in \cite{IBA_sync} also hold in partially asynchronous algorithmic model introduced in \cite{AA_convergence_markov}.

In the second part of the report, we explore the problem in asynchronous networks. While the traditional Byzantine consensus is not solvable in asynchronous systems \cite{FLP_one_crash}, approximate Byzantine consensus can be solved using iterative algorithms \cite{AA_Dolev_1986}.

\section{Preliminaries}
\label{sec:models}

In this section, we present the network and failure models that are common to both parts.

\subsection{Network Model}

The network is modeled as a simple {\em directed} graph $G(\scriptv,\scripte)$, where $\scriptv=\{1,\dots,n\}$ is the set of $n$ nodes, and $\scripte$ is the set of directed edges between nodes in $\scriptv$. With a slight abuse of terminology, we use the terms ``edge'' and ``link'' interchangeably.
We assume that $n\geq \max(2,3f+1)$, since the consensus problem for $n=1$ is trivial.
If a directed edge $(i,j)\in \scripte$, then node $i$ can reliably transmit to node $j$.
For convenience,
we exclude self-loops from $\scripte$, although every node is allowed to send messages
to itself. We also assume that all edges are authenticated, such that when a node $j$ receives a message from node $i$ (on edge $(i,j)$), it can correctly determine that the message was sent by node $i$.
For each node $i$, let $N_i^-$ be the set of nodes from which $i$ has incoming
edges.
That is, $N_i^- = \{\, j ~|~ (j,i)\in \scripte\, \}$.
Similarly, define $N_i^+$ as the set of nodes to which node $i$
has outgoing edges. That is, $N_i^+ = \{\, j ~|~ (i,j)\in \scripte\, \}$.
By definition, $i\not\in N_i^-$ and $i\not\in N_i^+$.
However, we emphasize that each node can indeed send messages to itself.

\comment{======================== Some more explanation==================
Last, it is convenient to explicitly divide the state of each node $i$ into two parts, memory holding $i$'s current value and an incoming buffer for each incoming edge $(j, i)$, for all $j \in N_i^-$. As implicitly assumed in \cite{AA_Dolev_1986}, we discuss the algorithms that have infinite-size buffer and automatically order the incoming messages. In other words, we consider a system on top of a reliable underlying communication layer. The layer adds a header containing a sequence number to each message while transmitting, and order messages based on the number when receiving. Then the layer returns the oldest message to the algorithms. In others words, we focus on the problem due to unknown delay: the confusion between faulty nodes and slow fault-free nodes. 

Though we do not consider link failure (i.e., message loss) in this report and in our previous work \cite{IBA_sync}, we do notice that to some degree, asynchronous iterative Byzantine consensus without link failure behaves similarly to the synchronous iterative Byzantine consensus with bounded dynamic link failures\footnote{By bounded dynamic link failures, we mean the following: At any given point of time, there are some link failures in the system, and the failures may change from link to link over time. However, the number of link failures is bounded.}.
==================================}

\subsection{Failure Model}

We consider the Byzantine failure model, with up to $f$ nodes becoming faulty. A faulty node may misbehave arbitrarily. Possible misbehavior includes sending incorrect and mismatching messages to different neighbors. The faulty nodes may potentially collaborate with each other. Moreover, the faulty nodes are assumed to have a complete knowledge of the state of the other nodes in the system and a complete knowledge of specification of the algorithm.

\newpage

\vspace*{4in}

\begin{center}
\textit{Part I: Synchronous Networks}
\end{center}

\newpage

\section*{Synchronous Networks}

The network is assumed to be synchronous. This report provides a more intuitive condition that is equivalent to our original necessary and sufficient condition introduced in Theorem 1 of \cite{IBA_sync}. Note that the discussion in this part is not self-contained, and relies heavily on the material and notations in \cite{IBA_sync}.

\section{More Intuitive Necessary and Sufficient Condition}
\label{s_tight_2}

For completeness, we state the tight condition from our previous report \cite{IBA_sync} here again:

\begin{theorem}
\label{thm:nc}
Suppose that a correct IABC algorithm exists for $G(\scriptv,\scripte)$.
Let sets $F,L,C,R$ form a partition\footnote{Sets $X_1,X_2,X_3,...,X_p$
are said to form a partition of set $X$ provided that (i) $\cup_{1\leq i\leq p} X_i = X$,
and (ii) $X_i\cap X_j=\Phi$ when $i\neq j$.} of $\scriptv$, such that $L$ and $R$ are both
non-empty, and $F$ contains at most $f$ nodes.
Then, at least one of these two conditions must be true: (i) $C\cup R\Rightarrow L$, or (ii) $L\cup C\Rightarrow R$.\footnote{Note that the notion of ``$\Rightarrow$'' and ``$\ARightarrow$'' (will be introduced in asynchronous networks part) is similar to ``r-robust" graph presented in \cite{IBA_broadcast_Sundaram}.}
\end{theorem}

This condition is not very intuitive.
In Theorem~\ref{thm:nc2} below, we state another tight necessary and sufficient condition that is equivalent
to the necessary condition in Theorem~\ref{thm:nc}, and is somewhat
easier to interpret. To facilitate the
statement of Theorem~\ref{thm:nc2}, we now introduce the notions of ``source
component'' and ``reduced graph'' using the following three definitions.

\begin{definition}
\label{def:decompose}
{\bf Graph decomposition:}
Let $H$ be a directed graph. Partition graph $H$ into strongly connected components,
$H_1,H_2,\cdots,H_h$, where $h$ is a non-zero integer dependent on graph $H$,
such that
\begin{itemize}
\item every pair of nodes {\bf within} the same strongly connected component has directed
paths in $H$ to each other, and
\item for each pair of nodes, say $i$ and $j$, that belong to
two {\bf different} strongly connected components, either $i$ does not have a
directed path to $j$ in $H$, or $j$ does not have a directed path to $i$ in $H$.
\end{itemize}
Construct a graph $H^d$ wherein each strongly connected component $H_k$ above is represented
by vertex $c_k$, and there is an edge from vertex $c_k$ to vertex $c_l$ only if
the nodes in $H_k$ have directed paths in $H$ to the nodes in $H_l$.
\end{definition}

~

It is known that the decomposition
graph $H^d$ is a directed {\em acyclic} graph \cite{dag_decomposition}.

~

\begin{definition}
{\bf Source component}:
Let $H$ be a directed graph, and let $H^d$ be its decomposition as per
Definition~\ref{def:decompose}. 
Strongly connected component $H_k$ of $H$ is said to be a {\em source component}
if the corresponding vertex $c_k$ in $H^d$ is \underline{not} reachable from any
other vertex in $H^d$. 
\end{definition}

\begin{definition}
\label{def:reduced} {\bf Reduced Graph:}
For a given graph $G(\scriptv,\scripte)$ and $F\subset\scriptv$,
a graph $G_F(\scriptv_F,\scripte_F)$
is said to be a {\em reduced graph}, if: (i)
$\scriptv_F=\scriptv-F$, and (ii)
$\scripte_F$ is obtained by first removing from $\scripte$ all the links
incident on the nodes in $F$, and {\em then} removing up to $f$ other incoming
links at each node in $\scriptv_F$.
\end{definition}

~

Note that for a given $G(\scriptv,\scripte)$ and a given $F$,
multiple reduced graphs $G_F$ may exist.

~

\begin{theorem}
\label{thm:nc2}
Suppose that Theorem \ref{thm:nc} holds for graph $G(\scriptv,\scripte)$.
Then, for any $F \subset \scriptv$ such that $|F|<|\scriptv|$ and $|F| \leq f$, every
reduced graph $G_F$ obtained as per Definition \ref{def:reduced} 
must contain exactly one {\em source component}.
\end{theorem}

\begin{proof}
Since $|F|<|\scriptv|$, 
$G_F$ contains at least one node;
therefore, at least one
source component must exist in $G_F$. We now prove that $G_F$ cannot
contain more than one source component. The proof is by contradiction.
Suppose that there exists a set $F\subset \scriptv$ with
$|F|<|\scriptv|$ and $|F|\leq f$,
and a reduced graph 
$G_F(\scriptv_F,\scripte_F)$ corresponding to $F$, such
that the decomposition of $G_F$ includes at least two source components.

Let the sets of nodes in two such source components of $G_F$
be denoted $L$ and $R$, respectively. Let $C=\scriptv-F-L-R$.
Observe that $F,L,C,R$ form a partition of the nodes in $\scriptv$.
Since $L$ is a source component in $G_F$ it follows that
there are no directed links in $\scripte_F$ from any node in
$C\cup R$ to the nodes in $L$.
Similarly, since $R$ is a source component in $G_F$ it follows that
there are no directed links in $\scripte_F$ from any node in $L\cup C$ to
the nodes in $R$.
These observations, together with the manner in which $\scripte_F$
is defined, imply that (i) there are at most $f$ links in $\scripte$ from
the nodes in $C\cup R$ to each node in $L$, and
(ii) there are at most $f$ links in $\scripte$ from
the nodes in $L\cup C$ to each node in $R$.
Therefore, in graph $G(\scriptv,\scripte)$, $C\cup R\not\Rightarrow L$
and $L\cup C\not\Rightarrow R$, violating Theorem~\ref{thm:nc}.
Thus, we have proved that $G_F$ must contain exactly one source component. 
\end{proof}

~

The above proof shows that Theorem \ref{thm:nc} implies Theorem \ref{thm:nc2}. Now, we prove that Theorem \ref{thm:nc2} implies Theorem \ref{thm:nc}.

\begin{proof}
Suppose that the condition stated in
Theorem \ref{thm:nc} does not hold for $G(\scriptv,\scripte)$. Thus, there exists a partition $F, L, C, R$ of $\scriptv$ such that $|F|\leq f$, $L$ and $R$ are non-empty,
and $C \cup R \not\Rightarrow L$ and $L \cup C \not\Rightarrow R$. 

We now construct a reduced graph $G_F(\scriptv_F,\scripte_F)$ corresponding to set $F$.
First, remove all nodes in $F$ from $\scriptv$ to obtain $\scriptv_F$.
Remove all the edges incident on $F$ from $\scripte$.
Then because $C \cup R \not\Rightarrow L$, the number of incoming
edges at each node in $L$ from  the nodes in $C\cup R$ is at most
$f$; remove all these edges.
Similarly, for every node $j \in R$, remove all incoming edges from $L \cup C$
(there are at most $f$ such edges at each node $j\in R$).
The resulting graph $G_F$ is a reduced graph that satisfies the conditions
in Definition \ref{def:reduced}.

In $\scripte_F$, there are no incoming edges to nodes
in $R$ from the nodes $L\cup C$; similarly, in $\scripte_F$, there
are no incoming edges to nodes $L$ from the nodes in $C\cup R$. 
It follows that no single node in $\scriptv_F$ has paths in $G_F$
(i.e., paths consisting of links in $\scripte_F$)
to all the other nodes in $\scriptv_F$. Thus, $G_F$ must contain more than one source component. Thus,
Theorem \ref{thm:nc2} does not hold for $G(\scriptv,\scripte)$.
\end{proof}

By two results above, it follows that Theorems \ref{thm:nc} and \ref{thm:nc2} specify
equivalent conditions.\footnote{An alternate interpretation of the condition
in Theorem \ref{thm:nc2} is that in graph $G_F$ non-fault-tolerant
iterative consensus must be possible.} 

Next, we present a weaker necessary conditions derived from Theorem \ref{thm:nc2} that implies the property of the source component.

~

\begin{corollary}
Suppose that Theorem \ref{thm:nc} holds for graph $G(\scriptv,\scripte)$.
Then, for any $F \subset \scriptv$ such that $|F| \leq f$, the unique
source component in every
reduced graph $G_F$ 
must contain at least $f+1$ nodes.
\end{corollary}
\begin{proof}
The proof is by contradiction.
Suppose that there exists a set $F$ with $|F|\leq f$,
and a corresponding reduced graph
$G_F(\scriptv_F,\scripte_F)$, such that
the decomposition of 
$G_F$ contains a unique source component
consisting of at most $f$ nodes.
Define $L$ to be the set of nodes in this unique source
component. Also define $C=\Phi$ and $R=\scriptv-L-F-C$.
Observe that $F,L,C,R$ form a partition of $\scriptv$.

Since $|L \cup C| = |L| \leq f$, it follows that in graph $G(\scriptv,\scripte)$,
$L \cup C \not\Rightarrow R$,
Then Theorem~\ref{thm:nc} implies that, in graph $G(\scriptv,\scripte)$,
$C\cup R \Rightarrow L$. That is, since $C=\Phi$,
$R\Rightarrow L$, and there must be a node in $L$, say node $i$,
that has at least $f+1$ links in $\scripte$ from the nodes in $R$. 
Since $i\in L$, it follows that $i\not\in F$ (by definition of $\Rightarrow$). Also, since
$i$ has at least $f+1$ incoming edges in $\scripte$ from nodes in $R$, 
it follows that in $\scripte_F$, node $i$ must have at least one incoming edge
from the nodes in $R$. This contradicts that assumption that set
$L$ containing node $i$ is a source component of $G_F$.
\end{proof}

~

Note that this Corollary implies that for the correctness of IABC on the graph, the graph must have a component that acts as a source with at least $f+1$ nodes and thus outnumbers the faulty nodes.

For a ``local'' fault model under the constraint that fault nodes send
{\em identical} messages to their outgoing neighbors, Zhang and Sundaram \cite{IBA_broadcast_Sundaram} showed {\em sufficiency} of a graph property similar to the condition above, although they do not
prove that the sufficient condition is also necessary. Also, our fault model does
not impose the above constraint on the faulty nodes.

\section{Partially Asynchronous Algorithmic Model}

\cite{AA_convergence_markov} (Chapter 7) presents a {\em Partially Asynchronous
Algorithmic Model}, in which
an iterative algorithm analogous to Algorithm 1 \cite{IBA_sync} is used to solve iterative consensus with zero faults, with the following
modifications:

\begin{itemize}
\item Each node may not necessarily update its state in each iteration. However, each node updates its state at least once in each set of consecutive $B$
iterations, where $B$ is a finite positive integer constant and is known to all nodes in advance.

\item If node $i$ updates its state in iteration $t$, due to message
delays, node $i$
may not necessarily be aware of the most recent state (i.e., at the
end of the previous iteration) of its incoming neighbors. However,
node $i$ will know the state of each incoming neighbor at the end of
at least one of the $B$ previous iterations\footnote{If node $i$ does not receive new values from some incoming neighbor $j$ in the past $B$ consecutive iterations, then by the model definition, node $i$ knows $j$ is faulty.}; the most recent state
known is used in performing state update at node $i$.

\end{itemize}

We believe that the necessary and sufficient
conditions for the IABC algorithm under partially asynchronous algorithmic model are identical to the necessary and
sufficient conditions presented above and in \cite{IBA_sync} for the synchronous model. We expect that the proof is similar to the proof presented in \cite{IBA_sync}.

\newpage

\vspace*{4in}

\begin{center}
\textit{Part II: Asynchronous Networks}
\end{center}

\newpage

\section*{Asynchronous Networks}

In this part, we consider the iterative consensus problem in asynchronous networks. We will follow the definition of asynchronous system used in \cite{AA_Dolev_1986}. Each node operates at a completely arbitrary rate. Furthermore, the link between any pair of nodes suffers from an arbitrary but finite network delay\footnote{The delay can also be variable.} and out-of-order delivery. 

Now, we introduce the class of algorithms that we will explore in this report.

\section{Asynchronous Iterative Approximate Byzantine Consensus}

\paragraph{Algorithm Structure}

By the definition of asynchronous systems, each node proceeds at different rate. Thus, Dolev et al. developed an algorithm based on ``rounds" such that nodes update once in each round \cite{AA_Dolev_1986}. In particular, we consider the structure of {\em Async-IABC Algorithm} below, which has the same structure as the algorithm in \cite{AA_Dolev_1986}. This algorithm structure differs 
from the one for synchronous systems in \cite{IBA_sync} in two important ways:
(i) the messages containing states are now tagged by the
round index to which the states correspond, and 
(ii) each node $i$ waits to receive only $|N_i^-|-f$ messages containing
states from round $t-1$ before computing the new state in
round $t$.

Due to the asynchronous nature of the system, different nodes may
potentially perform their $t$-th round at very different real times. Thus, the main difference between iteration and round is as following:

\begin{itemize}
\item Iteration is defined as fixed amount of real-time units. Hence, every node will be in the same iteration at any given real time.
\item Round is defined as the time that each node updates its value\footnote{With a slight abuse of terminology, we will use ``value'' and ``state'' interchangeably in this report.}. Hence, every node may be in totally different rounds at any given real time in asynchronous systems.
\end{itemize}

In Async-IABC algorithm, each node $i$ maintains state $v_i$, with $v_i[t]$ denoting the state of node $i$ at the end of its $t$-th round. Initial state of node $i$, $v_i[0]$, is equal to the initial input provided to node $i$. At the start of the $t$-th round ($t > 0$), the state of node $i$ is $v_i[t-1]$. Now, we describe the steps that should be performed by each node $i\in \scriptv$ in its $t$-th round. 

\vspace*{8pt}\hrule
{\bf Async-IABC Algorithm}
\vspace*{6pt}\hrule

\begin{enumerate}
\item {\em Transmit step:} Transmit current state $v_i[t-1]$ on all outgoing edges. The message
is tagged by index $t-1$.
\item {\em Receive step:} Wait until the first $|N_i^-|-f$ messages tagged by index $t-1$ are received
on the incoming edges (breaking ties arbitrarily). Values received in these messages form
vector $r_i[t]$ of size $|N_i^-|-f$.

\item {\em Update step:} Node $i$ updates its state using a transition function $Z_i$.

 $Z_i$ is a part of the specification of the algorithm, and takes
 as input the vector $r_i[t]$ and state $v_i[t-1]$.

\begin{eqnarray}
v_i[t] & = &  Z_i ~( ~r_i[t]\,,\,v_i[t-1] ~)
\label{eq:Z_i_async}
\end{eqnarray}

\end{enumerate}
\hrule
\vspace*{8pt}

We now define $U[t]$ and $\mu[t]$, assuming that $\scriptf$
is the set of Byzantine faulty nodes, with the nodes
in $\scriptv-\scriptf$ being non-faulty.\footnote{\normalsize For sets $X$ and $Y$, $X-Y$ contains elements that are in $X$ but not in $Y$. That is, $X-Y=\{i~|~ i\in X,~i\not\in Y\}$.} 
\begin{itemize}

\item $U[t] = \max_{i\in\scriptv-\scriptf}\,v_i[t]$. $U[t]$ is the largest state among the fault-free nodes at the end of the $t$-th round.
Since the initial state of each node is equal to its input,
$U[0]$ is equal to the maximum value of the initial input at the fault-free nodes.

\item $\mu[t] = \min_{i\in\scriptv-\scriptf}\,v_i[t]$. $\mu[t]$ is the smallest state among the fault-free nodes at the end of the $t$-th round.
$\mu[0]$ is equal to the minimum value of the initial input at the
fault-free nodes.
\end{itemize}

The following conditions must be satisfied by an Async-IABC algorithm
in the presence of up to $f$ Byzantine faulty nodes:
\begin{itemize}
\item {\em Validity:} $\forall t>0,
~~\mu[t]\ge \mu[t-1]
~\mbox{~~and~~}~
~U[t]\le U[t-1]$

\item {\em Convergence:} $\lim_{\,t\rightarrow\infty} ~ U[t]-\mu[t] = 0$
\end{itemize}

The objective in this report is to identify the necessary and sufficient
conditions for the existence of a {\em correct} Async-IABC algorithm (i.e.,
satisfying the above validity and convergence conditions)
for a given $G(\scriptv,\scripte)$ in any asynchronous system.

\subsection{Notations}

There are many notations used and will be introduced later in this part of the report. Here is a quick reference:

\begin{itemize}
\item $N_i^+, N_i^-$: set of outgoing neighbors and incoming neighbors of some node $i$, respectively.
\item $U[t], \mu[t]$: maximum value and minimum value of all the fault-free nodes at the end of round $t$, respectively.
\item $Z_i$: a function specifying how node $i$ updates its new value (algorithm specification).
\item $N_i^@[t]$: set of incoming neighbors from whom node $i$ actually received values at round $t \geq 1$.
\item $r_i[t]$: set of values sent by $N_i^@[t]$.
\item $N_i^*[t]$: set of incoming neighbors from whom node $i$ actually used the values to update at round $t \geq 1$.
\end{itemize}

Note that by definition we have the following relationships: $N_i^*[t] \subset N_i^@[t] \subset N_i^-$. Moreover, $N_i^*[t]$ and $N_i^@[t]$ may change over the rounds, and $N_i^-$ is a constant. Lastly, $|N_i^@[t]| = |N_i^-| - 2f$ and $|N_i^*[t]| = |N_i^@[t]| - f$ for any round $t \geq 1$.

\section{Necessary Condition}
\label{s_necessary}

In asynchronous systems, for an Async-IABC algorithm satisfying the
 the {\em validity} 
 and {\em convergence} conditions
 to exist, the underlying graph $G(\scriptv,\scripte)$
must satisfy a necessary condition proved in this section.
We now define relations $\ARightarrow$
and $\not\ARightarrow$ that are used frequently in our proofs. Note that these definitions are analogous to the definitions of $\Rightarrow$ and $\not\Rightarrow$ in \cite{IBA_sync}.

\begin{definition}
\label{def:absorb}
For non-empty disjoint sets of nodes $A$ and $B$,

\begin{itemize}
\item $A \ARightarrow B$ iff there exists a node $v\in B$ that has at least $2f+1$ incoming
links from nodes in $A$, i.e., $|N_v^-\cap A|>2f$.

\item $A\not\ARightarrow B$ iff $A\ARightarrow B$ is {\em not} true.

\end{itemize}

\end{definition}

\dividerline

Now, we present the necessary condition for correctness of Async-IABC in asynchronous systems. Note that it is similar to that for synchronous systems \cite{IBA_sync}, but with $\Rightarrow$ replaced by $\ARightarrow$.

\begin{theorem}
\label{thm:nc_a}
Let sets $F,L,C,R$ form a partition of $\scriptv$, such that
\begin{itemize}
\item $0 \leq |F|\le f$,
\item $0<|L|$, and
\item $0<|R|$
\end{itemize}
Then, at least one of the two conditions below must be true.
\begin{itemize}
\item $C\cup R\ARightarrow L$
\item $L\cup C\ARightarrow R$
\end{itemize}
\end{theorem}
\begin{proof}
The proof is by contradiction.
Let us assume that a correct Async-IABC consensus algorithm exists,
and $C\cup R\not\ARightarrow L$ and $L\cup C\not\ARightarrow R$.
Thus, for any $i\in L$, $|N_i^-\cap (C\cup R)|<2f+1$,
and for any
$j\in R$, $|N_j^-\cap (L\cup C)|<2f+1$,

Also assume that the nodes in $F$ (if $F$ is non-empty) are all faulty,
and the remaining nodes, in sets $L,R,C$, are fault-free. Note that the fault-free nodes
are not necessarily aware of the identity of the faulty nodes.

Consider the case when (i) each node in $L$ has input $m$, (ii) each
node in $R$ has input $M$, such that $M>m$,
and (iii) each node in $C$, if $C$ is non-empty,
has an input in the range
$[m,M]$.

At the start of round 1, suppose that the faulty nodes in $F$ (if non-empty)
send $m^- < m$ to outgoing neighbors in $L$, send $M^+ > M$ to outgoing neighbors in $R$, and
send some arbitrary value in $[m,M]$ to outgoing neighbors in $C$ (if $C$ is
non-empty).
This behavior is possible since nodes in $F$ are faulty.
Note that $m^-<m<M<M^+$.
Each fault-free node $k\in\scriptv-\scriptf$, sends to nodes
in $N_k^+$ value $v_k[0]$ in round 1.

Consider any node $i \in L$. Denote $N_i' = N_i^- \cap (C\cup R)$.
Since $C\cup R\not\ARightarrow L$, $|N_i'|\leq 2f$.
Consider the situation where the delay between certain $w = \min(f, |N_i'|)$ nodes in $N_i'$ and node $i$ is arbitrarily large compared to all the other traffic (including messages from incoming neighbors in $F$). Consequently, $r_i[1]$ includes $|N_i'| - w \leq f$ values from $N_i'$, since $w$ messages from $N_i'$ are delayed and thus ignored by node $i$. Recall that $N_i^@[1]$ is the set of nodes whose round $1$ values are received by node $i$ in time (i.e., before $i$ finishes step 2 in Async-IABC). By the argument above, $N_i^@[1] \cap N_i' \leq f$. 

Node $i$ receives $m^-$ from the nodes in $F \cap N_i^@[1]$, values in $[m,M]$ from the nodes in $N_i' \cap N_i^@[1]$, and $m$ from the nodes in $\{i\}\cup (L\cap N_i^@[1])$. 

Consider four cases:
\begin{itemize}
\item $F \cap N_i^@[1]$ and $N_i' \cap N_i^@[1]$ are both empty:
In this case, all the values that $i$ receives are from nodes in $\{i\}\cup( L\cap N_i^@[1])$,
and are identical to $m$. By validity condition, node $i$ must set its
new state, $v_i[1]$, to be $m$ as well.

\item $F \cap N_i^@[1]$ is empty and $N_i' \cap N_i^@[1]$ is non-empty:
In this case, since $|N_i' \cap N_i^@[1]|\leq f$, from $i$'s perspective,
it is possible that all the nodes in
$N_i^@[1] \cap N_i'$ are faulty, and the rest of the nodes are fault-free. In this
situation, the values sent to node $i$ by the fault-free nodes (which are
all in $\{i\}\cup (L\cap N_i^@[1]))$ are all $m$, and therefore, $v_i[1]$ must be set to $m$
as per the validity condition.

\item $F \cap N_i^@[1]$ is non-empty and $N_i' \cap N_i^@[1]$ is empty:
In this case, since $|F \cap N_i^@[1] |\leq f$, it is possible that all the nodes in $F \cap N_i^@[1] $ are faulty,
and the rest of the nodes are fault-free. In this
situation, the values sent to node $i$ by the fault-free nodes (which are
all in $\{i\}\cup (L\cap N_i^@[1]))$ are all $m$, and therefore, $v_i[1]$ must be set to $m$
as per the validity condition.

\item Both $F \cap N_i^@[1]$ and $N_i' \cap N_i^@[1]$ are non-empty:
From node $i$'s perspective, consider two possible scenarios:
(a) nodes in $F \cap N_i^@[1]$ are faulty, and the other
nodes are fault-free, and (b) nodes in $N_i' \cap N_i^@[1]$ are faulty, and the
other nodes are fault-free.

In scenario (a), from node $i$'s perspective, the non-faulty nodes have values
in $[m,M]$ whereas the faulty nodes have value $m^-$. According to the validity
condition, $v_i[1] \geq m$. On the other hand, in scenario (b), the
non-faulty nodes have values $m^-$ and $m$, where $m^-<m$; so $v_i[1] \leq m$, according to
the validity condition. Since node $i$ does not know whether the
correct scenario is (a) or (b), it must update its state to satisfy the
validity condition in both cases. Thus, it follows that $v_i[1] = m$.
\end{itemize}
Observe that in each case above $v_i[1]=m$ for each node $i\in L$.
Similarly, we can show that $v_j[1] = M$ for each node $j \in R$.

Now consider the nodes in set $C$, if $C$ is non-empty.
All the values received by the nodes in $C$ are in $[m,M]$, therefore,
their new state must also remain in $[m,M]$, as per the validity condition.

The above discussion implies that, at the end of the first iteration,
the following conditions hold true: (i) state of each node in $L$ is
$m$, (ii) state of each node in $R$ is $M$, and (iii) state of each node
in $C$ is in $[m,M]$. These conditions are identical to the initial conditions
listed previously. Then, by induction, it follows that for
any $t \geq 0$, $v_i[t] = m, \forall i \in L$, and $v_j[t] = M, \forall j \in R$.
Since $L$ and $R$ contain fault-free nodes, the convergence requirement
is not satisfied. This is a contradiction to the assumption that a correct
Async-IABC algorithm exists.
\end{proof}

\begin{corollary}
\label{cor:nc2}
Let $\{F,L,R\}$ be a partition of $\scriptv$, such that $0\leq |F|\le f$, and
$L$ and $R$ are non-empty. Then, either $L\ARightarrow R$ or $R\ARightarrow L$.
\end{corollary}
\begin{proof}
The proof follows by setting $C=\Phi$ in Theorem~\ref{thm:nc_a}.
\end{proof}

\begin{corollary}
\label{cor:5f}
The number of nodes $n$ must exceed $5f$ for
the existence of a correct Async-IABC algorithm that tolerates $f$ failures.
\end{corollary}
\begin{proof}
The proof is by contradiction.
Suppose that $2\leq n\leq 5f$, and consider the following two cases:
\begin{itemize}
\item $2\leq n\leq 4f$: Suppose that $L,R,F$ is a partition of $\scriptv$
such that $|L|=\lceil n/2 \rceil\leq 2f$,
$|R|=\lfloor n/2 \rfloor\leq 2f$ and $F=\Phi$. Note that $L$ and $R$
are non-empty, and $|L|+|R|=n$.
\item $4f<n\leq 5f$:

Suppose that $L,R,F$ is a partition of $\scriptv$,
such that $|L|=|R|=2f$ and $|F|=n-4f$. Note that
$0<|F|\leq f$.
\end{itemize}
In both cases above, Corollary~\ref{cor:nc2} is applicable. Thus,
either $L\ARightarrow R$ or $R\ARightarrow L$.
For $L\ARightarrow R$ to be true, $L$ must contain at least $2f+1$ nodes.
Similarly, for $R\ARightarrow L$ to be true, $R$ must contain at least
$2f+1$ nodes. Therefore, at least one of the sets $L$ and $R$ must contain more than
$2f$ nodes. This contradicts our choice of $L$ and $R$ above
(in both cases, size of $L$ and $R$ is $\leq 2f$).
Therefore, $n$ must be larger than $5f$. 
\end{proof}

\dividerline

\begin{corollary}
\label{cor:3f+1}
For the existence of a correct Async-IABC algorithm, then for each node $i\in\scriptv$,
$|N_i^-|\geq 3f+1$, i.e., each node $i$
has at least $3f+1$ incoming links, when $f>0$.
\end{corollary}
\begin{proof}
The proof is by contradiction. Consider the following two cases for some node $i$:

\begin{itemize}
\item $|N_i^-|\leq 2f$: Define set $F = \Phi, L = \{i\}$ and $R = V - F - L = V - \{i\}$. Thus, $N_i^-\cap R = N_i^-$, and $|N_i^-\cap R| \leq 2f$ by assumption.

\item $2f < |N_i^-| \leq 3f$: Define set $L = \{i\}$. Partition $N_i^-$ into two sets $F$ and $H$ such that $|F| = f$ and $|H| = |N_i^-| - f \leq 2f$. Define $R = V - F - L = V - F - \{i\}$. Thus, $N_i^-\cap R = H$, and $|N_i^-\cap R| \leq 2f$ by construction.
\end{itemize}

In both cases above, $L$ and $R$ are non-empty, so Corollary \ref{cor:nc2} is applicable. However, in each case, $L = \{i\}$ and $|L| = 1 < 2f+1$; hence, $L \not\ARightarrow R$. Also, since $L = \{i\}$ and  $|N_i^- \cap R| \leq 2f$, and hence $R \not\ARightarrow L$ by the definition of $\ARightarrow$. This leads to a contradiction. Hence, every node must have at least $3f+1$ incoming neighbors.

\end{proof}

\section{Useful Lemmas}
\label{s_useful}

In this section, we introduce two lemmas that are used in our proof of convergence. Note that the proofs are similar to corresponding lemmas in \cite{IBA_sync} except for the adoption of $\ARightarrow$ and ``rounds'' instead of $\Rightarrow$ and ``iterations.''

\begin{definition}
For disjoint sets $A,B$, 
$in(A \ARightarrow B)$ denotes the set of
all the nodes in $B$ that each have at least $2f+1$ incoming links from
nodes in $A$. More formally,
\[
in(A\ARightarrow B) = \{~v~|v\in B \mbox{\normalfont~and~}~2f+1\leq |N_v^-\cap A|~\}
\]

With a slight abuse of notation, when $A\not\ARightarrow B$, define $in(A\ARightarrow B)=\Phi$.
\end{definition}

\shortdividerline

\begin{definition}
\label{def:absorb_sequence}
For {\em non-empty disjoint} sets $A$ and $B$, set $A$ is said to {\bf propagate to} set $B$ in $l$ rounds, where $l>0$,
if there exist sequences of sets $A_0,A_1,A_2,\cdots,A_l$ and $B_0,B_1,B_2,\cdots,B_l$ (propagating sequences) such that

\begin{itemize}
\item $A_0=A$, $B_0=B$, $B_l=\Phi$, and, for $\tau<l$, $B_\tau \neq \Phi$.
\item for $0\leq \tau\leq l-1$,
\begin{list}{}{}
\item[*] $A_\tau\ARightarrow B_\tau$,
\item[*] $A_{\tau+1} = A_\tau\cup in(A_\tau\ARightarrow B_\tau)$, and
\item[*] $B_{\tau+1} = B_\tau - in(A_\tau\ARightarrow B_\tau)$
\end{list}
\end{itemize}
\end{definition}
Observe that $A_\tau$ and $B_\tau$ form a partition of $A\cup B$,
and for $\tau<l$, $in(A_\tau\ARightarrow B_\tau)\neq \Phi$.
Also, when set $A$ propagates to set $B$, length $l$ above is
necessarily finite. In particular, $l$ is upper bounded by $n-2f-1$, since set $A$ must
be of size at least $2f+1$ for it to propagate to $B$.
\dividerline

\comment{====================Some explanation =================
=====Notice that in the proof of the following useful lemmas, there is no notion of iteration involves. In other words, the delay does not affect the correctness of the proof. As long as message can be delivered in order and in finite amount of time, the proof follows. Thus, using the new definitions, we have the following two lemmas for asynchronous algorithm. 

There is one subtle concept worthy of some discussion. The propagating sequence is a global view of the system\footnote{Such concept is just for the analysis, and each node does not need to know the global view.}. In synchronous system in \cite{IBA_sync}, there is not much confusion, since global time is consistent with local time at each node. In asynchronous system, the propagating sequence is still a global view, but it is a view with respect to the notion of "round" instead of real time (measured by some external clock). For example, $\tau$ here means all the node values in round $\tau$, and should not be confused with some real time $\tau$. ===========================}

\begin{lemma}
\label{lemma:absorb_condition}
Assume that $G(\scriptv,\scripte)$ satisfies Theorem~\ref{thm:nc_a}.
Consider a partition $A,B,F$ of $\scriptv$ such that
$A$ and $B$ are non-empty, and $|F|\leq f$.
If $B \not\ARightarrow A$, then set $A$ propagates to set $B$.
\end{lemma}

\begin{proof}
Since $A,B$ are non-empty, and $B\not\ARightarrow A$, by
Corollary~\ref{cor:nc2}, we have $A\ARightarrow B$.

The proof is by induction.
Define $A_0=A$ and $B_0=B$.
Thus $A_0\ARightarrow B_0$ and $B_0\not\ARightarrow A_0$.
Note that $A_0$ and $B_0$ are non-empty.

\noindent
{\em Induction basis}: For some $\tau\geq 0$,
\begin{itemize}
\item for $0\leq k < \tau$, $A_k\ARightarrow B_k$, and $B_k\neq \Phi$,
\item either $B_\tau=\Phi$ or $A_\tau\ARightarrow B_\tau$,
\item for $0\leq k< \tau$,
$A_{k+1} = A_k \cup in(A_k\ARightarrow B_k)$, and
$B_{k+1}=B_k - in(A_k\ARightarrow B_k)$
\end{itemize}
Since $A_0\ARightarrow B_0$, 
the induction basis holds true for $\tau=0$.
 
\noindent
{\em Induction:}
If $B_\tau=\Phi$, then the proof is complete, since all
the conditions specified in Definition~\ref{def:absorb_sequence} are satisfied
by the sequences of sets $A_0,A_1,\cdots,A_\tau$ and $B_0,B_1,\cdots,B_\tau$.

Now consider the case when $B_\tau\neq \Phi$. By assumption,
$A_k\ARightarrow B_k$, for $0\leq k\leq \tau$.
Define $A_{\tau+1} = A_\tau \cup in(A_\tau\ARightarrow B_\tau)$ and $B_{\tau+1}=B_\tau - in(A_\tau\ARightarrow B_\tau)$.
Our goal is to prove that either $B_{\tau+1}=\Phi$ or $A_{\tau+1}\ARightarrow B_{\tau+1}$.
If $B_{\tau+1}=\Phi$, then the induction is complete. Therefore, now let us assume
that $B_{\tau+1}\neq \Phi$ and prove that $A_{\tau+1}\ARightarrow B_{\tau+1}$.
 We will prove this by contradiction.

Suppose that $A_{\tau+1}\not\ARightarrow B_{\tau+1}$.
Define subsets $L,C,R$ as follows: $L=A_0$, $C=A_{\tau+1}-A_0$ and $R=B_{\tau+1}$. Due to the manner in which $A_k$'s and $B_k$'s
are defined, we also have $C=B_0-B_{\tau+1}$.
Observe that $L,C,R,F$ form a partition of $\scriptv$, where $L,R$ are
non-empty, and the following relationships hold:
\begin{itemize}
\item $C\cup R = B_0$, and
\item $L\cup C=A_{\tau+1}$
\end{itemize}
Rewriting $B_0\not\ARightarrow A_0$ and $A_{\tau+1}\not\ARightarrow B_{\tau+1}$,
using the above relationships, we have, respectively,
\[
C\cup R\not\ARightarrow L,
\]
and
\[
L\cup C\not\ARightarrow R
\]
This violates the necessary condition in Theorem~\ref{thm:nc_a}.
This is a contradiction, completing the induction.

Thus, we have proved that, either (i) $B_{\tau+1}=\Phi$,
or (ii) $A_{\tau+1}\ARightarrow B_{\tau+1}$.
Eventually, for large enough $t$, $B_t$ will become $\Phi$, resulting
in the propagating sequences $A_0,A_1,\cdots, A_t$ and
$B_0,B_1,\cdots,B_t$, satisfying the conditions in Definition~\ref{def:absorb_sequence}.
Therefore, $A$ propagates to $B$. 
\end{proof}

\dividerline

\begin{lemma}
\label{lemma:must_absorb}
Assume that $G(\scriptv,\scripte)$ satisfies Theorem~\ref{thm:nc_a}.
For any partition $A,B,F$ of $\scriptv$, where $A,B$ are both non-empty,
and $|F|\leq f$,
at least one of the following conditions must be true:
\begin{itemize}
\item $A$ propagates to $B$, or
\item $B$ propagates to $A$
\end{itemize}
\end{lemma}

\begin{proof}
Consider two cases:
\begin{itemize}
\item $A\not\ARightarrow B$: Then by Lemma \ref{lemma:absorb_condition},
$B$ propagates to $A$, completing the proof.

\item $A\ARightarrow B$: In this case, consider two sub-cases:
\begin{itemize}
\item $A$ propagates to $B$: The proof in this case is complete.

\item $A$ does not propagate to $B$:
Thus, propagating sequences defined in Definition~\ref{def:absorb_sequence}
do not exist in this case. More precisely, there must exist $k>0$,
and sets $A_0,A_1,\cdots,A_k$ and $B_0,B_1,\cdots,B_k$,
such that:
\begin{itemize}
\item $A_0=A$ and $B_0=B$, and
\item for $0\leq i\leq k-1$,
\begin{list}{}{}
\item[o] $A_i\ARightarrow B_i$,
\item[o] $A_{i+1} = A_i\cup in(A_i\ARightarrow B_i)$, and
\item[o] $B_{i+1} = B_i - in(A_i\ARightarrow B_i)$.
\end{list}
\item $B_{k}\neq \Phi$ and $A_{k}\not\ARightarrow B_{k}$.
\end{itemize}
The last condition above violates the requirements for $A$ to propagate
to $B$.

Now $A_{k}\neq \Phi$, $B_k\neq \Phi$, and $A_k,B_k,F$ form
a partition of $\scriptv$. Since $A_{k}\not\ARightarrow B_{k}$,
by Lemma \ref{lemma:absorb_condition},
$B_k$ propagates to $A_k$.

Since $B_k\subseteq B_0 = B$, $A\subseteq A_k$, and $B_k$ propagates
to $A_k$, it should be easy to see that $B$ propagates to $A$.
 
\end{itemize}
\end{itemize}
\end{proof}

\section{Sufficient Condition}
\label{s_sufficiency}

\subsection{Algorithm 2}

We will prove that there exists an Async-IABC algorithm -- particularly
{\em Algorithm 2} below -- that satisfies
the {\em validity} and {\em convergence} conditions provided that the
graph $G(\scriptv,\scripte)$ satisfies the necessary condition in
Theorem~\ref{thm:nc_a}. This implies that the necessary condition
in Theorem~\ref{thm:nc_a} is also sufficient.

{\em Algorithm 2} has the three-step structure, and it is
similar to algorithms that were analyzed in prior
work as well \cite{AA_Dolev_1986,AA_async_PCN}
(although correctness of the algorithm
under the necessary condition in Theorem \ref{thm:nc_a}
has not been proved previously).

\vspace*{8pt}\hrule
{\bf Algorithm 2}
\vspace*{4pt}\hrule

\begin{enumerate}

\item {\em Transmit step:} Transmit current state $v_i[t-1]$ on all outgoing edges.
\item {Receive step:} Wait until receiving values on all but $f$ incoming edges. These values form
vector $r_i[t]$ of size $|N_i^-|-f$.\footnote{If more than $|N_i^-|-f$ values arrive at the same time, break ties arbitrarily.}

\item {\em Update step:}
Sort the values in $r_i[t]$ in an increasing order, and eliminate
the smallest $f$ values, and the largest $f$ values (breaking ties
arbitrarily).
 Let $N_i^*[t]$ denote the identifiers of nodes from
whom the remaining $N_i^- - 3f$ values were received, and let
$w_j$ denote the value received from node $j\in N_i^*$.
For convenience, define $w_i=v_i[t-1]$ to be the value node
$i$ ``receives'' from itself.  
Observe that
if $j\in \{i\}\cup N_i^*[t]$ is fault-free, then $w_j=v_j[t-1]$.

Define
\begin{eqnarray}
v_i[t] ~ = ~ Z_i(r_i[t],v_i[t-1]) ~ = ~\sum_{j\in \{i\}\cup N_i^*[t]} a_i \, w_j
\label{e_Z}
\end{eqnarray}
where
\[ a_i = \frac{1}{|N_i^-|+1-3f}
\] 
Note that $|N_i^*[t]| = |N_i^-| - 3f$, and $i\not\in N_i^*[t]$
because $(i,i)\not\in\scripte$.
The ``weight'' of each term on the right-hand side of
(\ref{e_Z}) is $a_i$, and these weights add to 1.
Also, $0<a_i\leq 1$.
For future reference, let us define $\alpha$ as:
\begin{eqnarray}
\alpha = \min_{i\in \scriptv}~a_i
\label{e_alpha}
\end{eqnarray}

\end{enumerate}

\hrule

\subsection{Sufficiency}

In Theorems \ref{thm:validity} and \ref{thm:convergence}
in this section, we
prove that Algorithm 2 satisfies {\em validity} and {\em convergence}
conditions, respectively, provided that $G(\scriptv,\scripte)$
satisfies the condition below, which matches the necessary
condition stated in Theorem \ref{thm:nc_a}.

~

\noindent{\bf Sufficient condition:}
{\em 
For every partition $F,L,C,R$ of $\scriptv$, such that $L$ and $R$ are both
non-empty, and $F$ contains at most $f$ nodes,
at least one of these two conditions is true:
(i) $C\cup R\ARightarrow L$, or (ii) $L\cup C\ARightarrow R$.
}

~

Note that the proofs below are similar to the ones for synchronous systems in \cite{IBA_sync}. The main differences are the following:

\begin{itemize}
\item We need to consider only values in $N_i^@[t]$ not in $N_i^-$. This is due to different step 2 between Algorithm 1 \cite{IBA_sync} and Algorithm 2.
\item We interpret $t$ as round index, rather than iteration index.
\end{itemize}

~

\begin{theorem}
\label{thm:validity}
Suppose that $G(\scriptv, \scripte)$ satisfies Theorem~\ref{thm:nc_a}.
Then Algorithm 2 satisfies the {\em validity} condition.
\end{theorem}

\begin{proof}
Consider the $t$-th round, and any fault-free node $i\in\scriptv-\scriptf$.
Consider two cases:
\begin{itemize}
\item
$f=0$: In (\ref{e_Z}), note that $v_i[t]$ is computed using
states from the previous round at node $i$ and other nodes.
By definition of $\mu[t-1]$ and $U[t-1]$, $v_j[t-1]\in [\mu[t-1],U[t-1]]$
for all fault-free nodes $j\in\scriptv-\scriptf$.
Thus, in this case, all the values used in computing $v_i[t]$ are in the
range $[\mu[t-1],U[t-1]]$.
Since $v_i[t]$ is computed as a 
weighted average of these values, $v_i[t]$ is also within
$[\mu[t-1],U[t-1]]$.

\item $f>0$: By Corollary~\ref{cor:3f+1}, $|N_i^-|\geq 3f+1$. Thus, $|N_i^@|\geq 2f+1$, and $|r_i[t]| \geq 2f+1$.
When computing set $N_i^*[t]$, the largest $f$ and smallest $f$ values
from $r_i[t]$ are eliminated. Since at most $f$ nodes are faulty,
it follows that, either (i) the values received from the faulty
nodes are all eliminated, or (ii) the values from the
faulty nodes that still remain are between values
received from two fault-free nodes. Thus, the remaining values in $r_i[t]$ are
all in the range $[\mu[t-1],U[t-1]]$. Also, $v_i[t-1]$ is 
in $[\mu[t-1],U[t-1]]$, as per the definition of $\mu[t-1]$ and $U[t-1]$.
Thus $v_i[t]$ is computed as a 
weighted average of values in $[\mu[t-1],U[t-1]]$, and, therefore,
it will also be in $[\mu[t-1],U[t-1]]$.
\end{itemize}
Since $\forall i\in\scriptv-\scriptf$, $v_i[t]\in [\mu[t-1],U[t-1]]$,
the validity condition is satisfied.
\end{proof}

\dividerline

Before proving the convergence of Algorithm 2, we first present three lemmas. In the discussion below, we assume that $G(\scriptv, \scripte)$ satisfies the sufficient condition.

\begin{lemma}
\label{lemma:psi}
Consider node $i\in\scriptv-\scriptf$.
Let $\psi\leq \mu[t-1]$. Then, for $j\in \{i\}\cup N_i^*[t]$,
\[
v_i[t] - \psi  \geq  a_i ~ (w_j - \psi)
\]
Specifically, for fault-free $j\in \{i\}\cup N_i^*[t]$,
\[
v_i[t] - \psi  \geq  a_i ~ (v_j[t-1] - \psi)
\]
\end{lemma}

\begin{proof}
In (\ref{e_Z}), for each $j\in N_i^*[t]$, consider two cases:
\begin{itemize}
\item Either $j=i$ or  $j\in N_i^*[t]\cap (\scriptv-\scriptf)$: Thus, $j$ is fault-free.
In this case, $w_j=v_j[t-1]$. Therefore,
$\mu[t-1] \leq w_j\leq U[t-1]$.
\item $j$ is faulty: In this case, $f$ must be non-zero (otherwise,
all nodes are fault-free).  
From Corollary~\ref{cor:3f+1}, $|N_i^-|\geq 3f+1$. Thus, $|N_i^@|\geq 2f+1$, and $|r_i[t]| \geq 2f+1$.
Then it follows that the smallest $f$
values in $r_i[t]$ that are eliminated in step 2 of Algorithm 2 contain the state of at least one fault-free node,
say $k$.
This implies that $v_k[t-1] \leq w_j$.
This, in turn, implies that
$\mu[t-1] \leq w_j.$
\end{itemize}
Thus, for all $j\in \{i\}\cup N_i^*[t]$, we have $\mu[t-1] \leq w_j$.
Therefore,
\begin{eqnarray}
w_j-\psi\geq 0 \mbox{\normalfont~for all~} j\in\{i\} \cup N_i^*[t]
\label{e_algo_1}
\end{eqnarray}
Since weights in Equation~\ref{e_Z} add to 1, we can re-write that equation
as,
\begin{eqnarray}
v_i[t] - \psi &=& \sum_{j\in\{i\}\cup N_i^*[t]} a_i \, (w_j-\psi) \\
\nonumber
&\geq& a_i\, (w_j-\psi), ~~\forall j\in \{i\}\cup N_i^*[t]  ~~~~~\mbox{\normalfont from (\ref{e_algo_1})}
\end{eqnarray}
For non-faulty $j\in \{i\}\cup N_i^*[t]$, $w_j=v_j[t-1]$, therefore,
\begin{eqnarray}
v_i[t] -\psi &\geq & a_i\, (v_j[t-1]-\psi)
\end{eqnarray}
\end{proof}

\dividerline

Similar to the above result, we can also show the following lemma:

\begin{lemma}
\label{lemma:Psi}
Consider node $i\in\scriptv-\scriptf$.
Let $\Psi\geq U[t-1]$. Then, for $j\in \{i\}\cup N_i^*[t]$,
\[
\Psi - v_i[t] \geq  a_i ~ (\Psi - w_j)
\]
Specifically, for fault-free $j\in \{i\}\cup N_i^*[t]$,
\[
\Psi - v_i[t] \geq  a_i ~ (\Psi - v_j[t-1])
\]
\end{lemma}

~

Then we present the main lemma used in proof of convergence. Note that below, we use parameter $\alpha$ defined in (\ref{e_alpha}). Recall that in (\ref{e_Z}) in Algorithm 2, $a_i > 0$ for all $i$, and thus, $\alpha > 0$.

~

\begin{lemma}
\label{lemma:bounded_value}
At the end of the $s$-th round, suppose that
the fault-free nodes in $\scriptv-\scriptf$ can be partitioned into
non-empty sets
$R$ and $L$ such that (i) $R$ propagates to $L$ in $l$ rounds,
and (ii) the states of nodes
in $R$ are confined to an interval of length $\leq \frac{U[s]-\mu[s]}{2}$.
Then, 
\begin{eqnarray}
U[s+l]-\mu[s+l] \leq \left(1-\frac{\alpha^l}{2}\right)(U[s] - \mu[s])
\label{e:convergence:1}
\end{eqnarray}
\end{lemma}

\begin{proof}
Since $R$ propagates to $L$, as 
per Definition~\ref{def:absorb_sequence},
there exist sequences of sets
$R_0,R_1,\cdots,R_l$ and $L_0,L_1,\cdots,L_l$, where
\begin{itemize}
\item $R_0=R$, $L_0=L$, $L_l=\Phi$, for $0\leq \tau<l$, $L_\tau \neq \Phi$, and
\item for $0\leq \tau\leq l-1$,
\begin{list}{}{}
\item[*] $R_\tau\ARightarrow L_\tau$,
\item[*] $R_{\tau+1} = R_\tau\cup in(R_\tau\ARightarrow L_\tau)$, and
\item[*] $L_{\tau+1} = L_\tau - in(R_\tau\ARightarrow L_\tau)$
\end{list}
\end{itemize}
Let us define the following bounds on the states of the nodes
in $R$ at the end of the $s$-th round:
\begin{eqnarray}
M & = & max_{j\in R}~ v_j[s] \\ \label{e_M}
m & = & min_{j\in R}~ v_j[s] \label{e_m}
\end{eqnarray}
By the assumption in the statement of Lemma~\ref{lemma:bounded_value},
\begin{eqnarray}
M-m\leq \frac{U[s]-\mu[s]}{2} \label{e_M_m}
\end{eqnarray}
Also, $M\leq U[s]$ and $m\geq \mu[s]$.
Therefore, $U[s]-M\geq 0$ and $m-\mu[s]\geq 0$.

The remaining proof of Lemma~\ref{lemma:bounded_value} relies
on derivation of the three intermediate claims below. 

\shortdividerline

\begin{claim}
\label{claim:1}
For $0\leq \tau\leq l$, for each node $i\in R_\tau$,
\begin{eqnarray}
v_i[s+\tau] - \mu[s] \geq  \alpha^{\tau}(m-\mu[s])
\label{e_ind_1}
\end{eqnarray}
\end{claim}

\noindent{\em Proof of Claim \ref{claim:1}:}
The proof is by induction.

{\em Induction basis:}
For some $\tau$, $0\leq \tau< l$, for each node $i\in R_\tau$,
(\ref{e_ind_1}) holds.
By definition of $m$, the induction basis holds true for $\tau=0$.

\noindent{\em Induction:}
Assume that the induction basis holds true for some $\tau$, $0\leq \tau<l$.
Consider $R_{\tau+1}$.
Observe that $R_\tau$ and $R_{\tau+1}-R_\tau$ form a partition of $R_{\tau+1}$;
let us consider each of these sets separately.
\begin{itemize}
\item Set $R_\tau$: By assumption, for each $i\in R_\tau$, (\ref{e_ind_1})
holds true.
By validity of Algorithm 2, $\mu[s] \leq \mu[s+\tau]$.
Therefore, setting $\psi=\mu[s]$ in Lemma~\ref{lemma:psi},
we get,
\begin{eqnarray*}
v_i[s+\tau+1] - \mu[s] & \geq 
& a_i~(v_i[s+\tau] - \mu[s]) \\
& \geq & a_i~ \alpha^{\tau}(m-\mu[s]) ~~~~~~~~ \mbox{due to (\ref{e_ind_1})} \\
& \geq & \alpha^{\tau+1}(m-\mu[s])  ~~~~~~~~~~ \mbox{due to (\ref{e_alpha})}
\end{eqnarray*}

\item Set $R_{\tau+1}-R_\tau$: Consider a node $i\in R_{\tau+1}-R_\tau$. By definition
of $R_{\tau+1}$, we have that $i\in in(R_\tau\ARightarrow L_\tau)$.
Thus,
\[ |N_i^- \cap R_\tau| \geq 2f+1 \] 

It follows that 

\[ |N_i^@[s+\tau] \cap R_\tau| \geq f+1 \] 

In Algorithm 2, $2f$ values ($f$ smallest and $f$ largest) received by
node $i$ are eliminated before $v_i[s+\tau+1]$ is computed at
the end of $(s+\tau+1)$-th round. Consider two possibilities:
\begin{itemize}
\item Value received from one of the nodes in $N_i^@[s+\tau] \cap R_\tau$ is
{\bf not} eliminated. Suppose that this value is received from
fault-free node $p\in N_i^@[s+\tau] \cap R_\tau$. Then, by an argument similar to the
previous case, we can set $\psi=\mu[s]$
in Lemma~\ref{lemma:psi}, to obtain,
\begin{eqnarray*}
v_i[s+\tau+1] -\mu[s] & \geq & a_i~(v_p[s+\tau]-\mu[s]) \\
& \geq & a_i~ \alpha^{\tau}(m-\mu[s]) ~~~~~~~~ \mbox{due to (\ref{e_ind_1})} \\
& \geq & \alpha^{\tau+1}(m-\mu[s])  ~~~~~~~~~~ \mbox{due to (\ref{e_alpha})}
\end{eqnarray*}

\item Values received from {\bf all} (there are at least $f+1$) nodes
in $N_i^@[s+\tau] \cap R_\tau$ are eliminated. Note that in this case $f$ must be
non-zero (for $f=0$, no value is eliminated, as already considered in the
previous case). By Corollary~\ref{cor:3f+1}, we know that
each node must have at least $3f+1$ incoming edges. Thus, $N_i^@[t+\tau] \geq 2f+1$.
Since at least $f+1$ values from nodes in $N_i^@[t+\tau] \cap R_\tau$ are
eliminated, and there are at least $2f+1$ values to choose
from, it follows that the values that are {\bf not} eliminated
are within the interval to which the values from $N_i^@[s+\tau] \cap R_\tau$ belong.
Thus, there exists a node $k$ (possibly faulty) from whom node $i$ receives
some value $w_k$ -- which is not eliminated -- and 
a fault-free node $p\in N_i^@[t+\tau] \cap R_\tau$ such that 
\begin{eqnarray}
v_p[s+\tau] &\leq & w_k \label{e_wk}
\end{eqnarray}
Then by setting $\psi=\mu[s]$ in Lemma~\ref{lemma:psi} we have
\begin{eqnarray*}
v_i[s+\tau+1] -\mu[s] & \geq & a_i~(w_k -\mu[s]) \\
& \geq & a_i~(v_p[s+\tau]-\mu[s]) ~~~~~~~~ \mbox{due to (\ref{e_wk})} \\
& \geq & a_i~ \alpha^{\tau}(m-\mu[s]) ~~~~~~~~ \mbox{due to (\ref{e_ind_1})} \\
& \geq & \alpha^{\tau+1}(m-\mu[s])  ~~~~~~~~~~ \mbox{due to (\ref{e_alpha})}
\end{eqnarray*}
\end{itemize}
\end{itemize}
Thus, we have shown that for all nodes in $R_{\tau+1}$,
\[
v_i[s+\tau+1] -\mu[s] 
\geq \alpha^{\tau+1}(m-\mu[s]) 
\]
This completes the proof of Claim \ref{claim:1}.

\shortdividerline

\begin{claim}
\label{claim:2}
For each node $i\in \scriptv-\scriptf$,
\begin{eqnarray}
v_i[s+l] - \mu[s] \geq  \alpha^{l}(m-\mu[s])
\label{e_ind_2}
\end{eqnarray}
\end{claim}

\noindent{\em Proof of Claim \ref{claim:1}:}

Note that by definition, $R_l = \scriptv-\scriptf$. Then the proof follows by setting $\tau = l$ in the above Claim \ref{claim:1}.

\comment{====== Notation replacement

$|N_i^- \cap R_{\tau}| \geq f + 1$ ==> $|N_i^@[s+\tau] \cap R_{\tau}| \geq 2f + 1$, so $|N_i^@[s+\tau] \cap R_{\tau}| \geq 2f + 1$

$p \in N_i^- \cap R_{\tau}|$ ==> $p \in N_i^@[s+\tau] \cap R_{\tau}|$

$N_i^- \cap R_{\tau}$ ==> $N_i^@[s+\tau] \cap R_{\tau}$

=====}

By a procedure similar to the derivation of Claim \ref{claim:2} above,
we can also prove the claim below. 

\begin{claim}
\label{claim:3}
For each node $i\in \scriptv-\scriptf$,
\begin{eqnarray}
U[s] - v_i[s+l] \geq  \alpha^{l}(U[s]-M)
\label{e_ind_3a}
\end{eqnarray}
\end{claim}

~

\noindent
Now let us resume the proof of the Lemma \ref{lemma:bounded_value}.
Note that $R_l=\scriptv-\scriptf$. Thus, 
\begin{eqnarray}
U[s+l] & = & \max_{i\in\scriptv-\scriptf}~ v_i[s+l] \nonumber \\
& \leq & U[s] - \alpha^{l}(U[s]-M) \mbox{~~~~~~~~~~~by (\ref{e_ind_3a})}
\label{e_U}
\end{eqnarray}
and
\begin{eqnarray}
\mu[s+l] & = & \min_{i\in\scriptv-\scriptf}~ v_i[s+l] \nonumber \\
& \geq & \mu[s] + \alpha^{l}(m-\mu[s]) \mbox{~~~~~~~~~~~by (\ref{e_ind_2}})
\label{e_mu}
\end{eqnarray}
Subtracting (\ref{e_mu}) from (\ref{e_U}),
\begin{eqnarray}
U[s+l]-\mu[s+l] & \leq & U[s] - \alpha^{l}(U[s]-M)  - \mu[s] - \alpha^{l}(m-\mu[s]) \nonumber \\
&=& (1-\alpha^l)(U[s]-\mu[s]) + \alpha^l(M-m) \\
&\leq& (1-\alpha^l)(U[s]-\mu[s]) + \alpha^l~\frac{U[s]-\mu[s]}{2}
		\mbox{~~~~~~~~~~by (\ref{e_M_m})} \\
&\leq& (1-\frac{\alpha^l}{2})(U[s]-\mu[s]) 
\end{eqnarray}
This concludes the proof of Lemma~\ref{lemma:bounded_value}.

\end{proof}

~

Now, we are able to prove the convergence of Algorithm 2. Note that this proof is essentially identical to the synchronous case \cite{IBA_sync}.

\begin{theorem}
\label{thm:convergence}
Suppose that $G(\scriptv, \scripte)$ satisfies Theorem~\ref{thm:nc_a}.
Then Algorithm 2 satisfies the {\em convergence} condition.
\end{theorem}

\begin{proof}

Our goal is to prove that, given any $\epsilon>0$, there
exists $\tau$ such that
\begin{equation}
U[t]-\mu[t] \leq \epsilon ~~~\forall t\geq \tau
\end{equation}

Consider the $s$-th round, for some $s\geq 0$.
If $U[s]-\mu[s]=0$, then the algorithm has already converged, and the proof
is complete, with $\tau=s$.

Now consider the case when $U[s]-\mu[s]>0$.
Partition $\scriptv-\scriptf$ into two subsets, $A$ and $B$, such
that, for each node $i\in A$, 
$v_i[s]\in \left[\mu[s], \frac{U[s]+\mu[s]}{2}\right)$, and
for each node $j\in B$,
$v_j[s] \in \left[\frac{U[s]+\mu[s]}{2}, U[s]\right]$.
By definition of $\mu[s]$ and $U[s]$, there exist fault-free nodes
$i$ and $j$ such that $v_i[s]=\mu[s]$ and $v_j[s]=U[s]$.
Thus, sets $A$ and $B$ are both non-empty.
By Lemma \ref{lemma:must_absorb}, one of the following two conditions
must be true:
\begin{itemize}
\item Set $A$ propagates to set $B$. Then, define $L=B$ and $R=A$.
The states of all the nodes in $R=A$ are confined within an
interval of length
$<\frac{U[s]+\mu[s]}{2} - \mu[s] \leq \frac{U[s]-\mu[s]}{2}$.

\item Set $B$ propagates to set $A$. Then, define $L=A$ and $R=B$.
In this case, states of all the nodes in $R=B$ are confined within an interval of length
$\leq U[s]-\frac{U[s]+\mu[s]}{2} \leq \frac{U[s]-\mu[s]}{2}$.

\end{itemize}
In both cases above, we have found non-empty sets $L$ and $R$
such that (i) $L,R$ is a partition of $\scriptv-\scriptf$,
(ii) $R$ propagates to $L$, and (iii) the states in $R$ are confined
to an interval of length $\leq \frac{U[s]-\mu[s]}{2}$.
Suppose that $R$ propagates to $L$ in $l(s)$ steps, where $l(s)\geq 1$.
By Lemma~\ref{lemma:bounded_value},
\begin{eqnarray}
U[s+l(s)]-\mu[s+l(s)] \leq \left( 1-\frac{\alpha^{l(s)}}{2}\right)(U[s] - \mu[s])
\label{e_t}
\end{eqnarray}
Since $n-f-1 \geq l(s)\geq 1$ and $0<\alpha\leq 1$, $0\leq \left( 1-\frac{\alpha^{l(s)}}{2}\right)<1$.

Let us define the following sequence of iteration indices\footnote{Without loss of generality, we assume that $U[\tau_i]-\mu[\tau_i] > 0$. Otherwise, the statement is trivially true due to the validity shown in Theorem \ref{thm:validity}.}:
\begin{itemize}
\item $\tau_0 = 0$,
\item for $i>0$, $\tau_i = \tau_{i-1} + l(\tau_{i-1})$, where $l(s)$ for any given $s$ was defined above.
\end{itemize}

By repeated application of the argument leading to (\ref{e_t}), we can prove
that, for $i\geq 0$,

\begin{eqnarray}
U[\tau_i]-\mu[\tau_i] \leq \left( \Pi_{j=1}^i\left( 1-\frac{\alpha^{\tau_j-\tau_{j-1}}}{2}\right)\right)~(U[0] - \mu[0])
\end{eqnarray}

For a given $\epsilon$,
by choosing a large enough $i$, we can obtain
\[
\left(\Pi_{j=1}^i\left( 1-\frac{\alpha^{\tau_j-\tau_{j-1}}}{2}\right)\right)~(U[0] - \mu[0]) \leq \epsilon
\]
and, therefore,
\begin{eqnarray}
U[\tau_i]-\mu[\tau_i] \leq  \epsilon
\end{eqnarray}
For $t\geq \tau_i$, by validity of Algorithm 1, it follows that
\[
U[t]-\mu[t] \leq
U[\tau_i]-\mu[\tau_i] \leq  \epsilon
\]
This concludes the proof.
\end{proof}

\section{Conclusion}

In this report, we present two sets of results. First, we prove another necessary and sufficient condition for the existence of synchronous IABC in arbitrary directed graphs. The condition is more intuitive  than the one in \cite{IBA_sync}. We also believe that the results can be extended to partially asynchronous algorithmic model presented in \cite{AA_convergence_markov}. In the second part, we extend our earlier results to asynchronous systems.


\begin{thebibliography}{1}

\bibitem{AA_async_PCN}
M.~H. Azadmanesh and R.~Kieckhafer.
\newblock Asynchronous approximate agreement in partially connected networks.
\newblock {\em International Journal of Parallel and Distributed Systems and
  Networks}, 5(1):26--34, 2002.

\bibitem{AA_convergence_markov}
D.~P. Bertsekas and J.~N. Tsitsiklis.
\newblock {\em Parallel and Distributed Computation: Numerical Methods}.
\newblock Optimization and Neural Computation Series. Athena Scientific, 1997.

\bibitem{dag_decomposition}
S.~Dasgupta, C.~Papadimitriou, and U.~Vazirani.
\newblock {\em Algorithms}.
\newblock McGraw-Hill Higher Education, 2006.

\bibitem{AA_Dolev_1986}
D.~Dolev, N.~A. Lynch, S.~S. Pinter, E.~W. Stark, and W.~E. Weihl.
\newblock Reaching approximate agreement in the presence of faults.
\newblock {\em J. ACM}, 33:499--516, May 1986.

\bibitem{FLP_one_crash}
M.~J. Fischer, N.~A. Lynch, and M.~S. Paterson.
\newblock Impossibility of distributed consensus with one faulty process.
\newblock {\em J. ACM}, 32:374--382, April 1985.

\bibitem{IBA_sync}
N.~H. Vaidya, L.~Tseng, and G.~Liang.
\newblock Iterative approximate byzantine consensus in arbitrary directed
  graphs.
\newblock {\em CoRR}, abs/1201.4183, 2012.

\bibitem{IBA_broadcast_Sundaram}
H.~Zhang and S.~Sundaram.
\newblock Robustness of information diffusion algorithms to locally bounded
  adversaries.
\newblock {\em CoRR}, abs/1110.3843, 2011.

\end{thebibliography}

\end{document}